\documentclass[12pt,a4paper]{article}
\usepackage[dvipdfmx,hiresbb]{graphicx}
\usepackage[dvipdfmx]{color}

\usepackage{color}
\usepackage{bm,enumerate,amsmath,amssymb,amsthm}
\usepackage{epsfig}
\usepackage{cite}
\usepackage{algorithm}
\usepackage{algorithmic}
\usepackage{txfonts}
\usepackage{subfigmat}
\usepackage{figsize}
\usepackage{here}
\usepackage{listings}
\usepackage{braket}
\usepackage{comment}
\usepackage{lscape}
\usepackage{bigints}

\usepackage{arxiv}
\usepackage[utf8]{inputenc} 

\theoremstyle{definition}
\newtheorem{theorem}{Theorem}
\newtheorem*{theorem*}{Theorem}

\newtheorem*{definition*}{Definition}

\newtheorem{remark}{Remark}

\newtheorem{assumption}{Assumption}

\renewcommand{\headeright}{}

\title{Lie-Algebraic Analysis of Generators: Approximation-Error Bounds and Barren-Plateau Heuristics}

\author{
  Hiroshi Ohno \\
  Toyota Central R \& D Labs., Inc.\\
  Aichi, Japan \\
  \texttt{oono-h@mosk.tytlabs.co.jp} \\
}

\date{\empty}

\begin{document}
\maketitle

\begin{abstract}
Lie algebras provide a useful framework for theoretical analysis in quantum machine learning, particularly in hybrid quantum-classical learning.
From the viewpoint of function approximation, expectation values of parameterized quantum circuits can be viewed as trigonometric polynomials whose accessible Fourier modes are determined by the spectra of the generators.
In this study, we describe:
(1) a minimax lower bound on the $ L^{2} $-approximation error over a Sobolev ball when the circuit's effective frequency set is contained in a radius-$K$ ball, which yields a scaling law of the form $ \Omega(K^{\frac{d}{2} - r}) $ for $ r > \frac{d}{2} $ (assuming the target function belongs to the Sobolev space $ W_2^{r}(\mathbb{T}^{d}) $),
and we also derive a Jackson-type upper bound on the approximation error of quantum circuits under Sobolev regularity of the target function, expressed in terms of an effective bandwidth determined by generator spectral gaps;
(2) a generator-selection rule motivated by enlarging the effective frequency set via non-commuting generators; and (3) a simple heuristic metric based on the trace component of generators, aimed at characterizing training behaviors related to barren plateaus.
Simulation experiments on toy problems illustrate the practical implications of the frequency-spectrum perspective and the proposed heuristics.
\end{abstract}

\section{Introduction}
Lie algebras are effective tools for analyzing the behavior of parameterized quantum circuit (unitary matrix) on training process.
Previous studies have demonstrated that Lie algebras and dynamical Lie algebras (DLAs) constitute a promising analytical framework for barren plateaus (BPs) \cite{larocca2025}, generalization performance, and expressive power of parameterized unitary matrices.
Furthermore, leveraging Lie algebras provides a framework for efficient classical simulation of quantum circuits \cite{goh2025}.

We provide a brief review of the related work.
In Ref. \cite{mcclean2018}, the authors showed that the gradient of the loss function vanishes exponentially with the number of qubits, and argued that this behavior is caused by the approximate 2-design property of random circuits.
When quantum circuits exhibit (approximate) Haar randomness, such vanishing gradients can arise, potentially leading to BPs during training.
Ragone et al. \cite{ragone2024} showed that the variance of the loss function is determined by the $ \mathfrak{g} $-purities of the measurement operator and the initial state $\rho$, as well as by the dimension of the DLA of the circuit's generators.
They further argued that when the DLA (a reductive Lie algebra) is abelian, the loss landscape is completely flat.
In other words, the DLA then consists only of its center.
Allcock et al. \cite{allcock2024} investigated the role of the center in a DLA for the Quantum Approximate Optimization Algorithm (QAOA).
For cycle and complete graphs, the presence of a nontrivial center implies that the variance of the loss function with respect to the parameters scales as $ \Theta(1) $. 
Consequently, the effect of BPs can vanish during QAOA training.

We speculate that the center in a DLA offers a promising perspective on barren plateaus.
However, practical design metrics for choosing generators that exploit the center are still lacking.

From the perspective of function approximation, parameterized quantum circuits can be regarded as a form of partial Fourier series.
In Ref. \cite{schuld2021}, the authors investigated the expressive power of parameterized quantum circuits as function approximators.
They proved that there exist quantum circuits capable of realizing arbitrary sets of Fourier coefficients, and that when the accessible frequency spectrum is asymptotically rich enough, such circuits become universal function approximators.

However, existing analyses of expressivity based on accessible Fourier modes typically do not provide a worst-case (minimax) guarantee on the approximation error for a given restriction of the spectrum.
In this study, we address this gap as follows. 
First, we formalize the expectation value of a parameterized unitary as a trigonometric polynomial whose effective frequency set $ S_{\rm eff} $ is determined by eigenvalue differences of the circuit generators.
Under a radius constraint $ \| u \| \le K $ on $ S_{\rm eff} $, we prove a minimax lower bound on the $ L^{2} $-approximation error over a Sobolev ball (Theorem \ref{th1}).
This result clarifies how spectral width (through $ K $) fundamentally limits the worst-case approximation capability of the circuit class.
Additionally, we derive a Jackson-type upper bound on approximation errors in terms of the circuit's effective bandwidth induced by generator spectra.
Second, motivated by the dependence of $ K $ (and thus the approximation limit) on the generator spectra, we propose a generator-selection rule based on enlarging $ S_{\rm eff} $, for example via non-commuting generators.
Third, from the Lie-algebraic perspective on training landscapes, we introduce a simple heuristic metric based on the trace component of generators and evaluate its behavior in toy experiments.

\section{Related Work}
\subsection{Quantum models as Fourier series}
A widely used perspective in quantum machine learning (QML) is that expectation values of parameterized quantum circuits can be written as Fourier-type sums, where the accessible frequencies are constrained by circuit structure and, in particular, by the spectra of generators.
Schuld et al. \cite{schuld2021} analyzed supervised quantum models as (partial) Fourier series and showed that the choice and repetition of encoding/processing gates control the frequency spectrum that the model can represent.
Subsequent works investigated how well common ans\"atze learn (truncated) Fourier series and compared architectures from an empirical standpoint \cite{heimann2025}.
Extensions to multidimensional settings have also been studied, clarifying how accessible frequency sets generalize beyond one-dimensional inputs \cite{casas2023}.
More recently, constraints on Fourier coefficients and phenomena such as vanishing expressivity have been reported, linking the decay of Fourier content to growing system size under certain architectures \cite{mhiri2025}.
These lines of work motivate analyzing expressivity and trainability through the lens of accessible spectral components.

\subsection{Fourier viewpoints and barren plateaus}
While the BP phenomenon is traditionally defined via the variance of loss gradients, Fourier-based analyses have also begun to connect spectral properties to trainability.
For example, Fourier-coefficient restrictions under BP-type conditions have been discussed, suggesting that limited Fourier energy can correlate with vanishing signals in parameterized circuits \cite{okumura2025}.
Such results complement the standard gradient-variance viewpoint by providing an alternative characterization in terms of spectral content.

\subsection{Lie-algebraic analyses of expressivity and trainability}
DLAs have emerged as a principled tool to characterize the reachable unitary families of variational quantum algorithms and to quantify expressivity and trainability.
A recent Lie-algebraic theory provides exact expressions for the variance of sufficiently deep circuit losses (and extensions to noise models), unifying multiple known sources of BPs within a single framework \cite{ragone2024}.
Within this program, the symmetry and representation-theoretic structure of the DLA plays a central role; in particular, reductive DLAs and nontrivial centers have been linked to the absence of BPs in structured ans\"atze.
For QAOA, analytical studies have derived explicit DLA decompositions including a direct-sum structure into a semisimple part and a low-dimensional center for certain highly symmetric graphs, together with closed-form variance expressions demonstrating the absence of BPs \cite{allcock2024}.

\subsection{Positioning of the present work}
Building on the above, our work focuses on connecting (i) the frequency-spectrum view of expectation values and (ii) Lie-algebraic structural information, with the goal of extracting practically interpretable design implications.
In contrast to prior Fourier-series expressivity studies that primarily emphasized representability of Fourier components \cite{schuld2021,heimann2025,casas2023}, we emphasize how generator spectra induce an effective frequency-radius constraint and how such a constraint yields a worst-case (minimax) approximation limit over smooth function classes.
In parallel to Lie-algebraic BP theories that predict loss/gradient concentration in deep regimes \cite{ragone2024}, we investigate simple generator-dependent heuristics intended to be computable directly from circuit specifications, and we validate their relationship to gradient-variance behavior in controlled experiments.

\section{Method}

\subsection{Approximation error of parameterized quantum circuits}
In this section, we describe the approximation error of quantum circuits in the context of Lie algebras.

Firstly, generator $ H $ can be decomposed using an eigenvalue decomposition as follows:
\begin{equation}
  H = \sum_{p} \lambda_{p} \ket{p}\bra{p}.
\end{equation}
Let us consider a unitary $ U $ with one parameter $ \theta $ as follows:
\begin{equation}
  U(\theta) = \exp{(-i \theta H)}.
\end{equation}
Then, the expectation $ f $ of $ U $ is as follows:
\begin{equation}
  f(\theta) = \braket{ \phi | U(\theta)^{\dagger} O U(\theta) | \phi }.
\end{equation}
Thus,
\begin{equation}
  \begin{split}
    U(\theta) &= \exp{(-i \theta H)} = \exp{\left( -i \theta \sum_{p} \lambda_{p} \ket{p}\bra{p} \right)}\\
    &= \sum_{p} \exp{( -i \theta \lambda_{p} )} \ket{p}\bra{p},
  \end{split}
\end{equation}
and
\begin{equation}\label{eq5}
  \begin{split}
    f(\theta) &= \braket{ \phi | U(\theta)^{\dagger} O U(\theta) | \phi }\\
    &= \braket{ \phi | \sum_{p} \exp{( i \theta \lambda_{p} )} \ket{p}\bra{p} O \sum_{q} \exp{( -i \theta \lambda_{q} )} \ket{q}\bra{q} \phi }\\
    &= \sum_{p, q} \rho_{p, q} O_{p, q} \exp{( -i \theta (\lambda_{q} - \lambda_{p}) )},
  \end{split}
\end{equation}
where
\begin{equation}
  \rho_{p, q} \coloneqq \braket{p|\phi}\braket{q|\phi},
\end{equation}
and
\begin{equation}
  O_{p, q} \coloneqq \braket{p|O|q}.
\end{equation}
For a circuit of depth $ L $ and block size $ m $, we denote the parameters by $ \theta = \{ \theta_{l,j} \}_{1 \le l \le L,\,1 \le j \le m} \in \mathbb{R}^{d}$ with $ d = L \cdot m $, and the function with depth $ L $ and size $ m $ is as follows:
\begin{equation}
  f(\theta) = \prod_{l=1}^{L} \prod_{j=1}^{m} \sum_{p,q} \rho_{p, q} O_{p, q} \exp{( -i \theta_{l,j} (\lambda^{(l,j)}_{q} - \lambda^{(l,j)}_{p}) )}.
\end{equation}
For each $ (l,j) $, let $ H_{l,j} $ be the corresponding generator and let $ \{ \lambda^{(l,j)}_{p} \}_{p} $ be its eigenvalues.
Define the spectral-gap set
\begin{equation}
  \Delta_{l,j} \coloneqq \{ \lambda^{(l,j)}_{q} - \lambda^{(l,j)}_{p} \; \mid \; p,q \}.
\end{equation}

Additionally, for each $ (l,j) $ define the maximum spectral width
\begin{equation}
  \omega^{(l,j)}_{\max} \coloneqq \max_{\omega \in \Delta_{l,j}} | \omega |,
  \qquad
  \omega_{\max} \coloneqq \max_{l,j} \omega^{(l,j)}_{\max}.
\end{equation}

In Eq. \ref{eq5}, the $ \theta $-dependence of $ f $ is governed by spectral gaps $ \lambda_{q} - \lambda_{p} $ of the generator(s), which are real-valued in general.
To connect this with the standard Fourier basis on $ \mathbb{T}^{d} $ (integer frequencies), we adopt the following mild normalization.

\begin{assumption}[Commensurate gaps / angle normalization]\label{assmp1}
  For each parameter $ \theta_{j} $ (corresponding to a generator $ H_{j} $), there exists a scale $ \gamma_{j} > 0 $ such that every spectral gap of $ H_{j} $ belongs to $ \gamma_{j} \mathbb{Z}$, i.e., $ \Delta_{j} \coloneqq \{ \lambda - \lambda' : \lambda, \lambda' \in \mathrm{spec}(H_{j}) \} \subset \gamma_{j}\mathbb{Z} $.
  Define the rescaled angles $ \varphi_{j} \coloneqq \gamma_{j} \theta_{j}$ and $ \varphi = (\varphi_{1},\dots,\varphi_{d}) $.
\end{assumption}

Under Assumption \ref{assmp1}, $ f(\theta) $ becomes $ 2\pi $-periodic in each $ \varphi_{j} $, hence can be regarded as a function on the torus $ \mathbb{T}^{d} \equiv [0, 2\pi]^{d} $, and admits the standard Fourier expansion with integer multi-indices.
In this study, the target function $ h $ belongs to the Sobolev space $ W_2^{r}(\mathbb{T}^{d}) $ with $ r > \frac{d}{2} $.
From here on, we work with the rescaled angles $ \varphi $ in Assumption \ref{assmp1} and identify the parameter domain with the torus $ \mathbb{T}^{d} = [0, 2\pi]^{d} $.
Accordingly, any $ h \in L^{2}(\mathbb{T}^{d})$ admits the Fourier series $ h(\varphi) = \sum_{s \in \mathbb{Z}^{d}} b_s \exp{(-is \cdot \varphi)} $, and we define $ W_{2}^{r}(\mathbb{T}^{d})$ in the usual manner.
Then, according to function approximation theory, we obtain the following theorem.
\begin{theorem} (Minimax approximation-error lower bound) \label{th1}
  Let $ d \in \mathbb{N} $ and $ r > \frac{d}{2} $.
  Define the Sobolev ball
  \begin{equation}
    \mathcal{B}_{r} \coloneqq \{ h \in W_{2}^{r}(T^{d}) \; \mid \; \| h \|_{W_{2}^{r}} \le 1\}.
  \end{equation}
  For $ K \ge 1$, define the class of trigonometric polynomials with Fourier support inside the radius-$ K $ $ \ell_{2} $-ball by
  \begin{equation}
    \mathcal{F}_{K} \coloneqq \left\{ f(\varphi) = \sum_{\| s \|_{2} \le K} a_{s} \exp{(-i s \cdot \varphi)} \; \mid \; \ a_{s} \in \mathbb{C} \right\}.
  \end{equation}
  Then there exists a constant $ c = c(d,r) > 0 $ such that
  \begin{equation}\label{eq18}
    \sup_{h \in \mathcal{B}_{r}}\ \inf_{f \in \mathcal{F}_{K}}\ \|f - h\|_{L^{2}(\mathbb{T}^{d})} \, \ge \, c \, K^{\frac{d}{2} - r} = c \, K^{-\alpha},
  \end{equation}
  where $ \alpha \coloneqq r - \frac{d}{2} > 0 $.
  Moreover, if the expectation value of a parameterized unitary satisfies $ S_{\rm eff} \subseteq \{u \in \mathbb{R}^{d} \; \mid \; \| u \|_{2} \le K \}$, then $ \mathcal{F}_{S_{\rm eff}} \subseteq \mathcal{F}_{K}$ and thus
  \begin{equation}\label{eq19}
    \sup_{h \in \mathcal{B}_{r}} \; \inf_{f \in \mathcal{F}_{S_{\rm eff}}} \; \|f - h\|_{L^{2}(\mathbb{T}^{d})} \, \ge \, c \, K^{-\alpha}.
  \end{equation}
  In particular, using $ K \le \sqrt{d} \, \omega_{\rm max}$ (or the rough bound $K \lesssim Lm \, \omega_{\rm max}$),
  \begin{equation}\label{eq20}
    \sup_{h \in \mathcal{B}_{r}} \; \inf_{f \in \mathcal{F}_{S_{\rm eff}}} \; \|f - h\|_{L^{2}(\mathbb{T}^{d})} \, \ge \, c \, (\sqrt{d} \, \omega_{\rm max})^{-\alpha}
    \quad \left(\text{or } \, \ge c \,(Lm \, \omega_{\rm max})^{-\alpha} \right).
  \end{equation}
\end{theorem}
\begin{proof}
  From the spectral decomposition, the expectation value of a one-parameter unitary contains only Fourier modes given by eigenvalue differences.
  Extending this to the multi-parameter setting, the expectation value $ f(\varphi) = \braket{ \phi | U(\varphi)^{\dagger} O U(\varphi) | \phi } $ is a finite trigonometric polynomial of the form
  \begin{equation}
    f(\varphi) = \sum_{u \in S_{\rm eff}} a_{u} \, \exp{(-i u \cdot \varphi)},
    \qquad a_{u} \in \mathbb{C}, \; \varphi \in \mathbb{T}^{d} = [0, 2\pi]^{d},
  \end{equation}
  where the effective frequency set $ S_{\rm eff} \subset \mathbb{R}^{d} $ satisfies
  \begin{equation}\label{eq12}
    S_{\rm eff} \subseteq \Delta_{1,1} \times \Delta_{1,2} \times \cdots \times \Delta_{L,m}.
  \end{equation}
  Then for any $ u \in S_{\rm eff} $ we have $ \| u \|_{2} \le \sqrt{d} \, \omega_{\rm max} $ and $ \| u \|_{1} \le d \, \omega_{\rm max} $, and hence there exists a radius $ K $ such that
  \begin{equation}\label{eq21}
    S_{\rm eff} \subseteq B_{2}(K),
    \qquad
    K \le \sqrt{d} \, \omega_{\max},
  \end{equation}
  where $ B_{2}(K) \coloneqq \{u \in \mathbb{R}^{d} \; \mid \; \| u \|_{2} \le K \} $.
  Equivalently, using the $ \ell_{1} $-radius gives the rough bound $ K \lesssim d \, \omega_{\rm max} = Lm \, \omega_{\rm max} $.

  Next, for a target function $ h \, : \, \mathbb{T}^{d} \to \mathbb{C} $, consider its Fourier expansion
  \begin{equation}
    h(\varphi) = \sum_{s \in \mathbb{Z}^{d}} b_{s} \exp{(-i s \cdot \varphi)},
    \qquad b_{s} \in \mathbb{C}, \; \varphi \in \mathbb{T}^{d} = [0, 2\pi]^{d}.
  \end{equation}
  We equip $ L^{2}(\mathbb{T}^{d}) $ with the inner product
  \begin{equation}
    \langle f, h \rangle \coloneqq \frac{1}{(2 \pi)^{d}} \int_{\mathbb{T}^{d}} f(\varphi) \overline{h(\varphi)} \, d \varphi,
    \qquad
    \| f \|_{L^{2}}^{2} \coloneqq \langle f, f\rangle,
  \end{equation}
  where $ \bar{h} $ denotes the complex conjugate of $ h $.
  By the Parseval equality,
  \begin{equation}
    \| h \|_{L^{2}(\mathbb{T}^{d})}^{2} = \sum_{s \in \mathbb{Z}^{d}} | b_{s} |^{2}.
  \end{equation}
  Moreover, $ \| h \|_{W_{2}^{r}} \le 1 $ implies
  \begin{equation}
    \sum_{s \in \mathbb{Z}^{d}}(1 + \| s \|_{2}^{2})^{r} \, | b_{s} |^{2} \le 1.
  \end{equation}
  For any $ K \ge 1 $, the $ L^{2} $-best approximation of $ h $ by $ \mathcal{F}_{K} $ is given by truncating its Fourier series to $ \{ \| s \|_{2} \le K\} $, and hence
  \begin{equation}\label{eq23}
    \inf_{f \in \mathcal{F}_{K}} \| f - h \|_{L^{2}(\mathbb{T}^{d})}^{2} = \sum_{\| s \|_{2} > K} | b_{s} |^{2}.
  \end{equation}

  To obtain a lower bound that is valid in the worst case over $\mathcal{B}_{r}$, we construct a function whose Fourier energy is concentrated on an annulus.
  Let $ A_{K} \coloneqq \{s \in \mathbb{Z}^{d} \; \mid \; K < \| s \|_{2} \le 2K \} $ and set
  \begin{equation}
    b_{s} \coloneqq
    \begin{cases}
      \displaystyle c_{0}\,(1 + \| s \|_{2}^{2})^{-\frac{r}{2}}, & s \in A_{K}, \\ 
      0, & \text{otherwise},
    \end{cases}
  \end{equation}
  where $ c_{0} > 0 $ is chosen so that $ \| h \|_{W_{2}^{r}} \le 1 $.
  Since $ | A_{K} | \asymp K^{d} $, we have
  \begin{equation}
    \sum_{s\in A_{K}}(1 + \| s \|_{2}^{2})^{r} \, | b_{s} |^{2} = c_{0}^{2} \, | A_{K} | \asymp c_{0}^{2} K^{d},
  \end{equation}
  and taking $ c_{0} \asymp K^{-\frac{d}{2}} $ ensures $ \| h \|_{W_{2}^{r}} \le 1 $, i.e., $ h \in \mathcal{B}_{r} $.
  Thus, for this $ h $, Eq. \ref{eq23} yields
  \begin{equation}
    \inf_{f \in \mathcal{F}_{K}} \| f - h \|_{L^{2}(\mathbb{T}^{d})}^{2} = \sum_{s \in A_{K}} | b_{s} |^{2} = c_{0}^{2} \sum_{s \in A_{K}}(1 + \| s \|_{2}^{2})^{-r} \asymp c_{0}^{2} \, |A_{K}| \, K^{-2r},
  \end{equation}
  where the definition of $ A_{K} $ is used.
  Using $ c_{0}^{2} \asymp K^{-d} $ and $ | A_{K} | \asymp K^{d} $, we obtain
  \begin{equation}
    \inf_{f \in \mathcal{F}_{K}} \| f - h \|_{L^{2}(\mathbb{T}^{d})}^{2} \, \ge \, c_{1} \, K^{-2r}
    \quad \Rightarrow \quad
    \inf_{f \in \mathcal{F}_{K}} \| f - h \|_{L^{2}(\mathbb{T}^{d})} \, \ge \, c \, K^{\frac{d}{2} - r},
  \end{equation}
  for $ c_{1} > 0 $ and some constant $ c = c(d,r) > 0 $ after adjusting constants using the lattice-point estimate on annuli\footnote{
    Here,  $ c = c(d, r) = c_{1} K^{\frac{d}{2} + r} $.
  }.
  Taking the supremum over $ h \in \mathcal{B}_{r} $ proves Eq. \ref{eq18}.

  Finally, if $ S_{\rm eff} \subseteq \{ u \; \mid \; \| u \|_{2} \le K \} $, then $ \mathcal{F}_{S_{\rm eff}} \subseteq \mathcal{F}_{K} $, thus Eq. \ref{eq19} follows immediately, and Eq. \ref{eq20} follows from Eq. \ref{eq21}.
\end{proof}

\begin{remark}[When Assumption \ref{assmp1} holds]
  Assumption \ref{assmp1} holds, for example, when each $ H_{j} $ is a Pauli-type observable (or a sum of commuting Pauli strings) with eigenvalues on a rational grid, as is common in hardware-efficient ans\"atze and QAOA-style circuits.
  More generally, one may always rescale $ (\theta_{j}, H_{j}) \mapsto (\gamma \theta_{j}, H_{j}/\gamma) $ without changing the implemented unitary, so the assumption can be seen as fixing a convenient unit for the parameter domain.
\end{remark}

\begin{remark}[Effect of angle rescaling]
  The rescaling $ \varphi_{j} = \gamma_{j} \theta_{j} $ changes the numerical value of the band-limit parameter $ K $ by a constant factor.
  Throughout, $ K $ should be understood with respect to the normalized angles $ \varphi $.
\end{remark}

From Eq. \ref{eq12}, increasing $ S_{\rm eff} $ yields a richer spectrum and thus improves the expressive (approximation) power of the parameterized unitary.
Therefore, when selecting generators from the DLA, we should remove mutually commuting generators, because non-commuting generators enlarge the effective set through the Baker-Campbell-Hausdorff (BCH) expansion, i.e., non-commuting generators do not increase the number of direct products.
This rule serves as a design policy for choosing generators.

Next, for an upper bound on the approximation error, we obtain the following theorem, which provides a Jackson-type upper bound on the best-achievable approximation error in terms of an effective bandwidth associated with the accessible frequency spectrum of the circuit.
We will measure representable low-frequency coverage by $ K_{\rm cov} $, and when the discrete accessible set is hard to characterize, we may lower-bound $ K_{\rm cov} $ using relaxed envelop $ \tilde{S} $ by verifying $ \mathbb{Z}^{d} \cup B_{2}(K) \subset \mathbb{Z}^{d} \cup \tilde{S} $.

\begin{theorem} (Jackson-type upper bound with minimal coverage)\label{th2}
  Let $ F_{S_{\rm eff}} $ be the effective function class defined by the accessible frequency set $ S_{\rm eff} $.
  \begin{equation}
    F_{S_{\rm eff}} \coloneqq \Big\{ f(\varphi) = \sum_{u \in S_{\rm eff}} a_{u} \exp{(-i u \cdot \varphi)} \; \mid \; a_{u} \in \mathbb{C} \Big\}.
  \end{equation}
  Let the target function be $ h(\varphi) = \sum_{s \in \mathbb{Z}^{d}} b_{s} \exp{(-i s \cdot \varphi)} $ ($ \in W_{2}^{r}(\mathbb{T}^{d}) $ with $ r > \frac{d}{2} $).
  Define the coverage radius
  \begin{equation}
    K_{\rm cov} \coloneqq \sup{ \{ K \geq 0 \; \mid \; \mathbb{Z}^{d} \cup B_{2} (K) \subset S_{\rm eff} \} }.
  \end{equation}
  Equivalently, all lattice frequencies with $ \| s \|_{2} \leq K $ are guaranteed to be representable.
  Then, for any $ 0 \leq K \leq K_{\rm cov} $, there exists $ f_{K} \in F_{S_{\rm eff}} $ such that
  \begin{equation}
    \| h - f_{K} \|_{L^{2}(\mathbb{T}^{d})} \leq C K^{\frac{d}{2} - r} \; \| h \|_{W_{2}^{r}(\mathbb{T}^{d})} = C K^{- \alpha} \; \| h \|_{W_{2}^{r}(\mathbb{T}^{d})},
  \end{equation}
  where $ C > 0 $ and $ \alpha \coloneqq r - \frac{d}{2} > 0 $, and hence
  \begin{equation}
    \inf_{f \in \mathcal{F}_{S_{\rm eff}}} \| h - f \|_{L^{2}(\mathbb{T}^{d})} \leq C K^{-\alpha}_{\rm cov} \; \| h \|_{W_{2}^{r}(\mathbb{T}^{d})}.
  \end{equation}
\end{theorem}
\begin{proof}
  Any parameter-induced frequency vector can then be viewed component-wise as $ u_{l,j} \in \mathbb{R}^{d} $ with $ d = L \cdot m $, where each component is constrained by the corresponding gap structure.
  This motivates the product-form frequency envelope
  \begin{equation}
    S_{\rm gap} \coloneqq \prod_{l = 1}^{L} \prod_{j = 1}^{m} \Delta_{l,j} \subset \mathbb{R}^{d},
  \end{equation}
  and we take $ S_{\rm eff} \subset S_{\rm gap} $.
  Since $ S_{\rm gap} $ is discrete and may be sparse, for the purpose of deriving approximation upper bounds, it is convenient to introduce a continuous envelope set
  \begin{equation}
    \tilde{S} \coloneqq \prod_{l = 1}^{L} \prod_{j = 1}^{m} [ -\omega_{\rm max}^{(l,j)}, \omega_{\rm max}^{(l,j)} ] \subset \mathbb{R}^{d}.
  \end{equation}
  By construction, $ S_{\rm eff} \subset S_{\rm gap} \subset \tilde{S} $.
  We quantify the effective bandwidth by defining $ K $ as the radius of the smallest Euclidean ball enclosing $ \tilde{S} $ (equivalently, $ \| u \|_{2} \leq K $ for all $ u \in \tilde{S} $), which yields the bound
  \begin{equation}
    K = \left( \sum_{l = 1}^{L} \sum_{j = 1}^{m} (\omega_{\rm max}^{(l, j)})^2 \right)^{\frac{1}{2}} \leq \sqrt{L m} \; \omega_{\rm max}.
  \end{equation}
  In the sequel, we invoke truncation/approximation arguments on the integer lattice modes $ \mathbb{Z}^{d} \cup B_{2}(K) $ (or on $ \mathbb{Z}^{d} \cup \tilde{S} $), which provides a conservative but robust notion of accessible bandwidth controlled solely by the generator spectral widths.
  In the below, we approximate a target $ h $ by its low-frequency truncation supported within the bandwidth $ K $.
  For a target function $ h $, using the Fourier expansion, the target function is expressed as follows:
  \begin{equation}
    h(\varphi) = \sum_{s \in \mathbb{Z}^{d}} b_{s} \exp{( -i s \cdot \varphi )},
  \end{equation}
  where $ \varphi \in \mathbb{R}^{d} $.
  For a bandwidth parameter $ 0 < K \leq K_{\rm cov} $, define the low-frequency truncation
  \begin{equation}
    h_{K} (\varphi) \coloneqq \sum_{\| s \|_{2} \leq K} b_{s} \exp{(-i s \cdot \varphi)}.
  \end{equation}
  By the definition of $ K_{\rm cov} $, we have $ \mathbb{Z}^{d} \cup B_{2}(K) \subset S_{\rm eff} $, and hence $ h_{K} \in F_{S_{\rm eff}} $.
  Since $ h \in W^{r}_{2}(\mathbb{T}^{d}) $,
  \begin{equation}
    \| h \|^{2}_{W^{r}_{2}(\mathbb{T}^{d})} = \sum_{s \in \mathbb{Z}^{d}} (1 + \| s \|^{2})^{r} |b_{s}|^{2} < \infty.
  \end{equation}
  Thus,
  \begin{equation}
    \sum_{\| s \| > K} |b_{s}|^{2} \lesssim \sum_{\| s \| > K} \| s \|^{-2r}.
  \end{equation}
  Here, because the number of grid points of an annulus with the radius $ t $ is $ \asymp t^{d-1} $,
  \begin{equation}
    \sum_{\| s \| > K} \| s \|^{-2r} \lesssim \int_{K}^{\infty} t^{d-1} t^{-2r} \, dt = \int_{K}^{\infty} t^{d-1-2r} \, dt.
  \end{equation}
  Hence,
  \begin{equation}
    \int_{K}^{\infty} t^{d-1} t^{-2r} \, dt = \int_{K}^{\infty} t^{d-1-2r} \, dt =
    \left\{
    \begin{array}{ll}
      \frac{K^{d-2r}}{2r-d} & r > \frac{d}{2}\\
      \infty & r \leq \frac{d}{2}.
    \end{array}
    \right.
  \end{equation}
  Thus, since $ r > \frac{d}{2} $,
  \begin{equation}
    \sum_{\| s \| > K} | b_{s} |^{2} \lesssim K^{d - 2r}.
  \end{equation}
  Using the Parseval equality,
  \begin{equation}
    \| h - h_{K} \|_{W^{r}_{2}(\mathbb{T}^{d})}^{2} = \sum_{\| s \|_{2} > K} | b_{s} |^{2}.
  \end{equation}
  Therefore, the tail bound is $ \sum_{\| s \| > K} | b_{s} |^{2} \lesssim K^{d - 2r} $, and thus $ \| h - h_{K} \|_{W^{r}_{2}(\mathbb{T}^{d})} \lesssim K^{\frac{d}{2} - r} $.
  Finally, since $ h_{K} $ is supported on frequencies with $ \| s \| \leq K $, whenever these modes are representable by the circuit class (i.e., $ h_{K} \in F_{S_{\rm eff}} $), we have the upper bound
  \begin{equation}
    \inf_{f \in \mathcal{F}_{S_{\rm eff}}} \| h - f \|_{L^{2}(\mathbb{T}^{d})} \leq \| h - h_{K} \|_{L^{2}(\mathbb{T}^{d})} \lesssim K^{-\alpha} \lesssim K_{\rm cov}^{-\alpha} \| h \|_{W^{r}_{2}(\mathbb{T}^{d})},
  \end{equation}
  for any $ K \leq K_{\rm cov} $, $ h_{K} \in F_{S_{\rm eff}} $.
  This yields the stated Jackson-type upper bound.
\end{proof}

Regarding $ K_{\rm cov} $ introduced in Theorem \ref{th2}, the imposed restriction may be strong and thus difficult to satisfy in practical quantum circuits.

\begin{remark}
  Based on Theorem \ref{th2}, we consider the limit of the upper bound of approximation error as follows:
  \begin{equation}
    \lim_{r, d \rightarrow \infty} d^{-(r - \frac{d}{2})}.
  \end{equation}
  Since $ d \geq 1 $, $ 0 < d^{-(r - \frac{d}{2})} \leq 1 $ (thus, it does not converge to $ \infty $).
  From the above equation, we obtain the following equation:
  \begin{equation}
    \lim_{r, d \rightarrow \infty} \exp{ \left \{ -\left(r - \frac{d}{2}\right) \ln{ d } \right \} }.
  \end{equation}
  Then, due to the difference of the convergence speed $ r - \frac{d}{2} $ and $ \ln{ d } $ to infinity (and the assumption $ r - \frac{d}{2} > 0 $), we obtain the following limit:
  \begin{equation}
    \lim_{r, d \rightarrow \infty} \exp{ \left \{ -\left(r - \frac{d}{2}\right) \ln{ d } \right \} } = 
    \left\{
    \begin{array}{ll}
      0 & \left(r - \frac{d}{2}\right) \ln{ d } \rightarrow \infty\\
      1 & \left(r - \frac{d}{2}\right) \ln{ d } \rightarrow 0\\
      e^{ -C } & \left(r - \frac{d}{2}\right) \ln{ d } \rightarrow C\\
    \end{array}
    \right.,
  \end{equation}
  where $ C $ is a constant.
  Therefore, the limit of the approximation error is an indeterminate form.
  However, when $ r - \frac{d}{2} \geq D > 0 $ (where $ D $ is any positive constant), $ \left(r - \frac{d}{2}\right) \ln{ d } \rightarrow \infty $, and the limit becomes zero.

  This observation is consistent with the intuition that increasing the effective bandwidth (or low-frequency coverage) improves approximation capability, whereas training may still suffer from BPs when circuits approach (approximate) Haar-random behavior.
\end{remark}

\subsection{Effects of center on training behaviors}
A parameterized unitary matrix $ U(\theta) $ is given by generators $ H_{j} $ as follows:
\begin{equation}
  U(\theta) = \exp{ \left( -i \sum_{j=1}^{m} \theta_{j} H_{j} \right) }.
\end{equation}
Then, a DLA $ \mathfrak{g} $ is defined by
\begin{equation}
  \mathfrak{g} \coloneqq {\rm span}_{\mathbb{R}} \braket{ iH_{1}, \ldots, iH_{m} }_{Lie},
\end{equation}
where $ {\rm span}_{\mathbb{R}} (\mathcal{A}) $ is the set of all linear combinations of elements of $ \mathcal{A} $ with real coefficients, and $ \braket{\mathcal{S}}_{Lie} $ denotes the Lie closure, that is, the set obtained by repeatedly taking the nested commutators between the elements in $ \mathcal{S} $.
Generally, a DLA is decomposed as follows:
\begin{equation}
  \mathfrak{g} = \mathfrak{z} \oplus [\mathfrak{g}, \mathfrak{g}], 
\end{equation}
where $ \mathfrak{z} $ denotes a center (Abelian) and $ [\mathfrak{g}, \mathfrak{g}] $ denotes a semi-simple algebra \cite{ragone2024}.
Here, $ [\mathfrak{g}, \mathfrak{g}] $ is expressed by the direct sum of simple algebras $ \mathfrak{g}_{j} $, i.e., $ \oplus_{j} \mathfrak{g}_{j} $.
Simple algebra such as $ \mathfrak{su} $, which has no center and is non-commutative, has high expressive power and strongly tends to exhibit BPs.
A semi-simple algebra such as $ \mathfrak{su}(N) \oplus \mathfrak{su}(N) $, which is the direct sum of simple algebras, has moderate expressive power.
However, when the Hilbert space can be decomposed by the simple algebras and the gradients of the cost function remain in each subspace, BPs can be mitigated.
When center $ \mathfrak{z} $ exists (which corresponds to a flat space in the gradient space), although the center is unrelated to BPs, BPs can be mitigated because of the parameter space partitioning.
Such an algebra, named a reductive DLA, is expressed as follows:
\begin{equation}
  \mathfrak{g} = \mathfrak{u}(N) = \mathfrak{su}(N) \oplus i \mathbb{R} I.
\end{equation}
In summary, the relationship between a DLA and a BP is presented in Table \ref{tab2-1}.
\begin{table}[htb]
  \centering
  \caption{Relationship between DLA and BP}\label{tab2-1}
  \vspace{5pt}
  \begin{tabular}{ccc}\hline
    DLA & Expressive power & Possibility of BP\\ \hline
    Simple algebra & Strong & High\\
    Semi-simple algebra & Moderate & Low\\
    Reductive algebra & Relatively strong & Dependence on strength of center\\ \hline
  \end{tabular}
\end{table}
Note that the table provides a comparison of the expressive powers and possibilities of a BP of semi-simple and reductive algebras with those of simple algebra.

From the above discussion, regarding generators (reductive DLA) and BPs, an empirical equation can be presented as follows:
\begin{equation}
\eta = \frac{| \; {\rm Tr}[ G ] \; |}{\| \; G \; \|_{{\rm F}}},
\end{equation}
where $ G $ denotes a generator\footnote{
  The $ \mathfrak{g} $-purity \cite{ragone2024} cannot assess the strength of the center's contribution in a reductive DLA.
}.
Here, the computational complexity of $ \eta $ is $ \mathcal{O}(N^{2}) $ for the matrix $ G $ of size $ N $.
When $ \eta $ is kept small, the curvature of the parameter space becomes smooth, and then BPs may be mitigated.
This heuristic metric is not intended to directly assess BPs; that is, it does not provide a necessary and sufficient condition for the occurrence of BPs.
Instead, by assessing the flatness of the parameter space, it allows us to infer the potential presence of BPs during the training process.
We believe that $ \eta $ can serve as a design metric for generators.
We will conduct experiments to demonstrate the effects of $ \eta $.

\section{Experiments and results}
\subsection{Approximation}
In this experiment, we illustrate the practical implication of Theorem \ref{th1}: the accessible Fourier modes of an expectation-value model are determined by eigenvalue differences of the generators, and the worst-case approximation capability is fundamentally limited by the effective frequency radius $ K $, which in turn is controlled by the spectral widths of the generators. 
Although Theorem \ref{th1} is a minimax statement over a Sobolev ball (and does not directly predict optimization outcomes), it motivates the following hypothesis: if the generator spectra of the training model are substantially narrower than those of the target model, then the training model cannot adequately represent the target's higher-frequency components, leading to a larger approximation (and hence prediction) error.

Thus, we investigate how the matching between the eigenvalue spectra of the target model and the training model influences the training process.

\paragraph{Setup.}
Training data were prepared using the target model.
The maximum eigenvalue of the generators of the target model was set to 10.\footnote{
  To construct a generator $ H $ (Hermitian matrix), we use the spectral decomposition as follows:
  \begin{equation*}
    H = U \Lambda U^{\dagger}.
  \end{equation*}
  Here, $ U $ is a unitary matrix and $ \Lambda $ is a diagonal matrix.
  In the experiments, the diagonal elements of $ \Lambda $ are set from $ -b $ to $ b $ at equal intervals, where $ b $ is the maximum eigenvalue.
}
The number of qubits $ n $ was three, and the depth $ L $ was set to five.
To construct the generators (Hermitian matrices), unitary matrices whose size was eight were Haar-randomly sampled.
The number of generators was five (as equal to $ L $).
All parameter $ \theta $s were set to one.
The input data $ x $ of the training data was uniformly randomly sampled from the range $ [-1, 1] $.
The output data $ y $ was $ \braket{\psi \, | \, Z(0) \, | \, \psi} $, where $ \ket{\psi} = U(\theta) \, U_{enta} \, RY(x)^{\otimes n} \ket{0}^{\otimes n} $.
For classical data encoding, we used the rotation Y matrix $ RY $.
In order to construct $ U_{enta} $, we adopted the circuit-block (CB) configuration and controlled X gates (CNOT).
The number of training data was 1000.

The training model was constructed in the same way as the target model.
The output of the model was $ \braket{\varphi \, | \, Z(0) \, | \, \varphi} $, where $ \ket{\varphi} = U(\theta) \, U_{enta} \, RY(x)^{\otimes n} \ket{0}^{\otimes n} $.
Here, $ \theta $ was trainable parameters.
The maximum eigenvalue of the generators of the model was varied in the range $ \{0.1, 1, 10\} $.

The training experiments were performed 10 times with different random seeds.\footnote{
  Quantum circuits were implemented in PennyLane \cite{bergholm2018}, and noiseless and state-vector simulations were performed.
  The optimizer was Adam by JAX \cite{jax2018}.
  The learning rate was $ 0.00001 $.
  The maximum number of epochs was 500.
}

\paragraph{Result.}
Average root-mean-square-error (RMSE) results are shown in Figure \ref{fig3-1}.
\begin{figure}[htbp]
  \centering
  \begin{tabular}{c}
    \includegraphics[width=6cm]{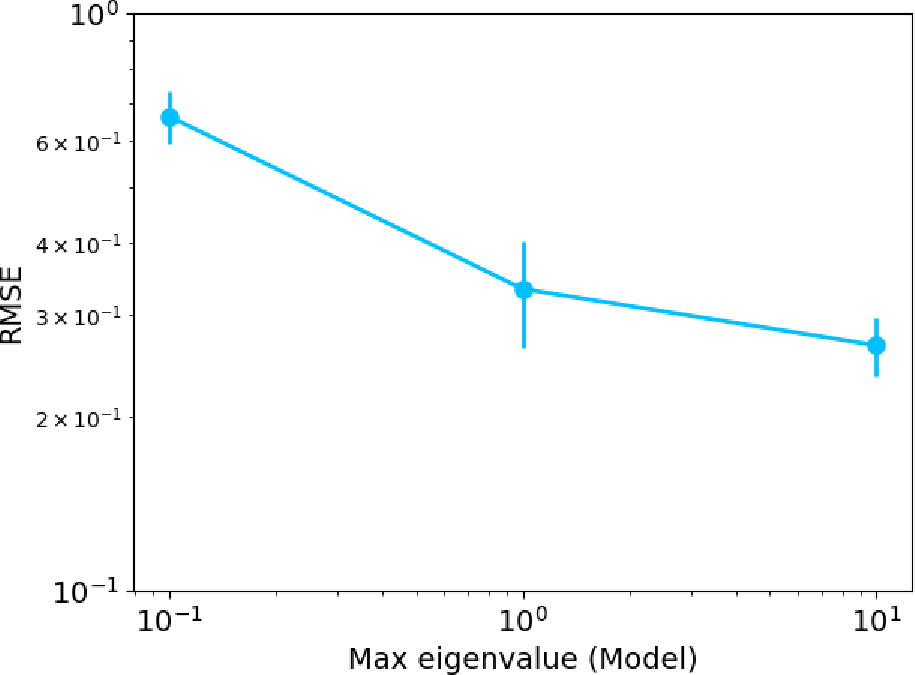}
  \end{tabular}
  \caption{RMSE result for target model with maximum eigenvalue 10 in accordance with maximum eigenvalue of model}\label{fig3-1}
\end{figure}
Error bars indicate one standard deviation.
As can be seen in the figure, the RMSE of the maximum eigenvalue 10 was the lowest.
There was a significant difference between the RMSE for 1.0 and that for 10.\footnote{
  The statistical test results were as follows: p-value = 0.0601 for the Shapiro-Wilk normality test, p-value = 0.0163 ($ < 0.05 $) for the t-test, and p-value = 0.0020 for the exact Wilcoxon signed-rank test.
  Here, the normality assumption was only weakly supported; therefore we performed the exact Wilcoxon signed-rank test.
}
This supports the spectrum-matching hypothesis motivated by Theorem \ref{th1}: enlarging and matching the accessible frequency range of the model improves its ability to approximate the target function represented by the target circuit.
Therefore, we believe that matching between the eigenvalue spectra of the target model and the training model is crucial in the training process.

\subsection{Variance of $ \partial_{\theta} C $}
To see the influence of the center, consider generator $ H $ composed of a semi-simple algebra $ SS $ and center $ Center $ as follows:
\begin{equation}
  H = weight \cdot SS + Center,
\end{equation}
where $ weight $ denotes a balancing factor between $ SS $ and $ Center $, being a value taken in the range $ [0, 1] $.
In the experiments, $ SS $ was $ I \otimes Y $ and $ Center $ was $ I \otimes I $, where $ I $ is the identity matrix and $ Y $ is the Pauli-Y matrix.

The unitary matrix $ U $ is given as follows:
\begin{equation}
  U(\theta) = \exp{(-i \theta H) }.
\end{equation}
Then, the derivative $ \partial_{\theta} U $ of $ U $ is $ -i H \, U(\theta) $.
Here, the cost function $ C $ is $ \braket{0^{\otimes n} \, | \, U^{\dagger} O U \, | \, 0^{\otimes n} } $ and its derivative $ \partial_{\theta} C $ with respect to $ \theta $ is $ \braket{0^{\otimes n} \, | \, \partial_{\theta} U O U + U^{\dagger} O \, \partial_{\theta} U \, | \, 0^{\otimes n} } $, where $ O $ denotes an observable.

\paragraph{Setup.}
The number of qubits $ n $ was two, and $ \theta $ was uniformly randomly sampled from the range $ [-2\pi, 2\pi] $.
$ O $ was $ I \otimes Z $.
$ \partial_{\theta} C $ was computed 50 times with different random seeds\footnote{
  Quantum circuits were implemented by NumPy \cite{harris2020}.
}.

\paragraph{Result.}
Figure \ref{fig3-2} shows the results of the variance of $ \partial_{\theta} C $ and $ \eta $ in accordance with $ weight $.
The error bars for the variance results indicate one standard deviation.
\begin{figure}[htbp]
  \centering
  \begin{tabular}{c}
    \includegraphics[width=6cm]{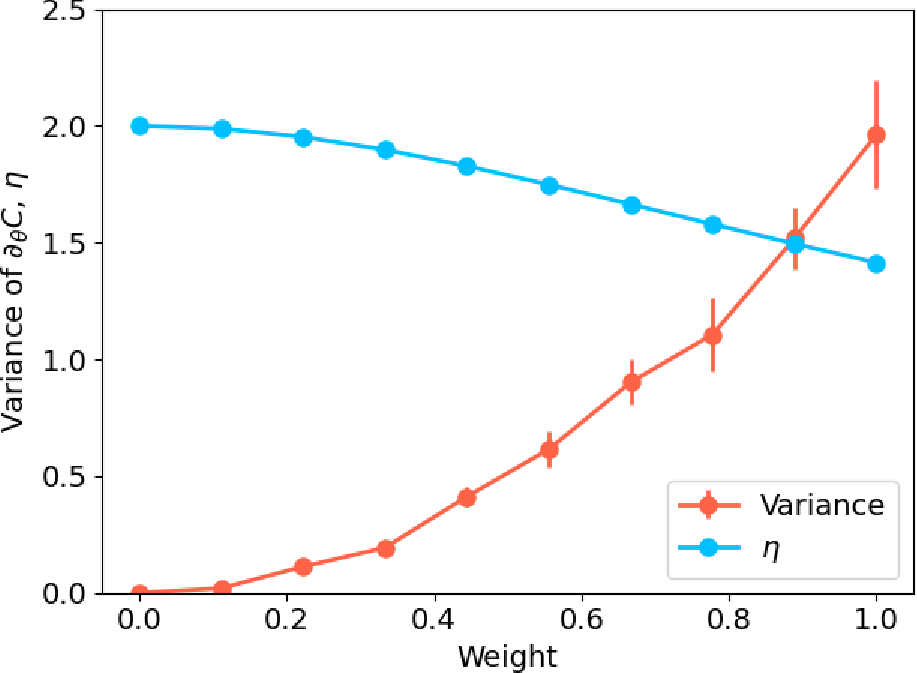}
  \end{tabular}
  \caption{Results of variance of $ \partial_{\theta} C $ and $ \eta $ in accordance with $ weight $}\label{fig3-2}
\end{figure}
When $ weight = 0 $, the Hamiltonian becomes $ H = I \otimes I $, and the variance becomes zero.
Increasing the weight value, i.e., relatively decreasing the importance of $ Center $, leads to an increase in the variance and a decrease in $ \eta $.
Therefore, with respect to the variance, $ \eta $ can serve as a metric for BPs.
Additionally, a decrease in $ \eta $ corresponds to an increase in $ SS $ (the semi-simple algebra), which enhances the expressive power of the unitary matrix.
Thus, we believe that, in generator design, $ \eta $ is a promising metric for assessing the potential presence of BPs.

\section{Conclusion}
Utilizing Lie algebras to construct generators, we presented lower and upper bounds on approximation errors and a generator design metric aimed at mitigating BPs and enhancing the expressive power of quantum circuits.
We obtained the following results:
\begin{itemize}
\item A lower bound on approximation errors and a Jackson-type upper bound on approximation errors of parameterized quantum circuits.
\item A generator selection rule derived from a DLA
\item A metric for inferring the potential occurrence of BPs
\end{itemize}

Theorems \ref{th1} and \ref{th2} are based on the assumption of an integer lattice.
Extending the analysis to non-harmonic Fourier series with real-valued frequency sets is left for future work.

Although the simulation experiments were conducted only on toy problems, the results support the validity of the proposed assessment metric.
To further verify the above findings, additional studies should be needed.
For example, regarding the assessment metric, we will investigate a recommended threshold value that balances BP mitigation and expressive power, as well as the effectiveness of a loss function regularized by the assessment metric.

\bibliographystyle{plain}
\bibliography{references}

\begin{thebibliography}{10}

\bibitem{allcock2024}
Jonathan Allcock, Miklos Santha, Pei Yuan, and Shengyu Zhang.
\newblock On the dynamical lie algebras of quantum approximate optimization
  algorithms.
\newblock {\em arXiv preprint arXiv:2407.12587}, 2024.

\bibitem{bergholm2018}
Ville Bergholm, Josh Izaac, Maria Schuld, Christian Gogolin, Carsten Blank,
  Keri McKiernan, and Nathan Killoran.
\newblock Pennylane: Automatic differentiation of hybrid quantum-classical
  computations.
\newblock {\em arXiv preprint arXiv:1811.04968}, 2018.

\bibitem{jax2018}
James Bradbury, Roy Frostig, Peter Hawkins, Matthew~James Johnson, Chris Leary,
  Dougal Maclaurin, George Necula, Adam Paszke, Jake Vander{P}las, Skye
  Wanderman-{M}ilne, and Qiao Zhang.
\newblock {JAX}: composable transformations of {P}ython+{N}um{P}y programs,
  2018.

\bibitem{casas2023}
Berta Casas and Alba Cervera-Lierta.
\newblock Multidimensional fourier series with quantum circuits.
\newblock {\em Physical Review A}, 107:062612, Jun 2023.

\bibitem{goh2025}
Matthew~L. Goh, Martin Larocca, Lukasz Cincio, M.~Cerezo, and Fr\'ed\'eric
  Sauvage.
\newblock Lie-algebraic classical simulations for quantum computing.
\newblock {\em Physical Review Research}, 7:033266, Sep 2025.

\bibitem{harris2020}
Charles~R. Harris, K.~Jarrod Millman, St{'{e}}fan~J. van~der Walt, Ralf
  Gommers, Pauli Virtanen, David Cournapeau, Eric Wieser, Julian Taylor,
  Sebastian Berg, Nathaniel~J. Smith, Robert Kern, Matti Picus, Stephan Hoyer,
  Marten~H. van Kerkwijk, Matthew Brett, Allan Haldane, Jaime~Fern{'{a}}ndez
  del R{'{\i}}o, Mark Wiebe, Pearu Peterson, Pierre G{'{e}}rard-Marchant, Kevin
  Sheppard, Tyler Reddy, Warren Weckesser, Hameer Abbasi, Christoph Gohlke, and
  Travis~E. Oliphant.
\newblock Array programming with {NumPy}.
\newblock {\em Nature}, 585(7825):357--362, Sep 2020.

\bibitem{heimann2025}
Dirk Heimann, Hans Hohenfeld, Gunnar Sch\"onhoff, Elie Mounzer, and Frank
  Kirchner.
\newblock Learning fourier series with parametrized quantum circuits.
\newblock {\em Physical Review Research}, 7:023151, May 2025.

\bibitem{larocca2025}
Mart{\'i}n Larocca, Supanut Thanasilp, Samson Wang, Kunal Sharma, Jacob
  Biamonte, Patrick~J. Coles, Lukasz Cincio, Jarrod~R. McClean, Zo{\"e} Holmes,
  and M.~Cerezo.
\newblock Barren plateaus in variational quantum computing.
\newblock {\em Nature Reviews Physics}, 7(4):174--189, Apr 2025.

\bibitem{mcclean2018}
Jarrod~R. McClean, Sergio Boixo, Vadim~N. Smelyanskiy, Ryan Babbush, and
  Hartmut Neven.
\newblock Barren plateaus in quantum neural network training landscapes.
\newblock {\em Nature Communications}, 9(1):4812, 2018.

\bibitem{mhiri2025}
Hela Mhiri, Leo Monbroussou, Mario Herrero-Gonzalez, Slimane Thabet, Elham
  Kashefi, and Jonas Landman.
\newblock Constrained and {V}anishing {E}xpressivity of {Q}uantum {F}ourier
  {M}odels.
\newblock {\em {Quantum}}, 9:1847, Sep 2025.

\bibitem{okumura2025}
Shun Okumura and Masayuki Ohzeki.
\newblock Fourier analysis of parameterized quantum circuits and the barren
  plateau problem.
\newblock {\em arXiv preprint arXiv:2309.06740}, 2025.

\bibitem{ragone2024}
Michael Ragone, Bojko~N. Bakalov, Fr{\'e}d{\'e}ric Sauvage, Alexander~F.
  Kemper, Carlos Ortiz~Marrero, Mart{\'i}n Larocca, and M.~Cerezo.
\newblock A {L}ie algebraic theory of barren plateaus for deep parameterized
  quantum circuits.
\newblock {\em Nature Communications}, 15(1):7172, Aug 2024.

\bibitem{schuld2021}
Maria Schuld, Ryan Sweke, and Johannes~Jakob Meyer.
\newblock Effect of data encoding on the expressive power of variational
  quantum-machine-learning models.
\newblock {\em Physical Review A}, 103:032430, Mar 2021.

\end{thebibliography}

\end{document}